\documentclass{amsart}

\usepackage[english]{babel}
\usepackage[utf8]{inputenc}
\usepackage[T1]{fontenc}

\usepackage{amsmath}
\usepackage{amssymb}
\usepackage{mathtools}
\usepackage{amsthm}
\usepackage{physics}
\usepackage{subfig}
\usepackage{cite}
\usepackage{braket}
\usepackage{xcolor}
\usepackage{booktabs}
\usepackage{tikz,tikz-3dplot}
\usepackage{pgfplots}
\usetikzlibrary{arrows}
\usetikzlibrary{patterns}

\newcommand{\IN}{\mathcal{IN}}

\newcommand{\Z}{\mathbb{Z}}
\newcommand{\N}{\mathbb{N}}
\newcommand{\R}{\mathbb{R}}
\newcommand{\K}{\mathbb{K}}

\usepackage{enumitem}
\newcounter{classcount}

\newcounter{listcount}
\renewcommand*\thelistcount{\arabic{listcount}}

\newcounter{eqlist}

\newcounter{eqlistI}
\renewcommand*\theeqlistI{E.\arabic{eqlistI}}
\newcounter{eqlistII}
\renewcommand*\theeqlistII{E.\arabic{eqlistII}\ensuremath{^\prime}}

\theoremstyle{plain}
\newtheorem{theorem}{Theorem}
\newtheorem{proposition}[theorem]{Proposition}

\newtheorem{conjecture}[theorem]{Conjecture}
\theoremstyle{definition}

\theoremstyle{remark}
\newtheorem{remark}{Remark}

\renewcommand{\imath}{\mathrm{i}}
\renewcommand{\epsilon}{\varepsilon}
\renewcommand{\phi}{\varphi}

\allowdisplaybreaks
\makeatletter 
\renewcommand{\pdv}[2]{\begingroup 
\@tempswafalse\toks@={}\count@=\z@ 
\@for\next:=#2\do 
{\expandafter\check@var\next\@nil
 \advance\count@\der@exp 
 \if@tempswa 
   \toks@=\expandafter{\the\toks@\,}%
 \else 
   \@tempswatrue 
 \fi 
 \toks@=\expandafter{\the\expandafter\toks@\expandafter\partial\der@var}}%
\frac{\partial\ifnum\count@=\@ne\else^{\number\count@}\fi#1}{\the\toks@}%
\endgroup} 
\def\check@var{\@ifstar{\mult@var}{\one@var}} 
\def\mult@var#1#2\@nil{\def\der@var{#2^{#1}}\def\der@exp{#1}} 
\def\one@var#1\@nil{\def\der@var{#1}\chardef\der@exp\@ne} 
\makeatother
\begin{document}
\title[Algebraic entropy of a class of differential-difference equations]{Algebraic entropy of a class of five-point differential-difference equations}

\author{Giorgio Gubbiotti}
\address{School  of  Mathematics  and  Statistics  F07,  The  University  of  Sydney,  NSW  2006, Australia}
\email{giorgio.gubbiotti@sydney.edu.au}

\date{\today}

\begin{abstract}
    We compute the algebraic entropy of a class of integrable
    Volterra-like five-point differential-difference equations 
    recently classified using the generalised symmetry method.
    We show that, when applicable, the results of the algebraic
    entropy agrees with the result of the generalised symmetry method,
    as all the equations in this class have vanishing entropy.
\end{abstract}

\maketitle

\section{Introduction}

One of the most important topic in modern Mathematical Physics is the
study of the so-called \emph{integrable systems}.
Roughly speaking integrable systems are important both from theoretical
and practical point of view since they can be regarded as \emph{universal models}
for physics going beyond the linear regime \cite{Calogero1991}.
The birth modern theory of integrable systems is usually recognized in the
seminal works of Zabusky and Kruskal \cite{Zabusky1965}, Gardner, Greene, Kruskal 
and Miura \cite{Gardner1967} and Lax \cite{Lax1968} on the Korteweg-deVries 
(KdV) equation \cite{Korteweg1895}.

The concept of integrability come from Classical Mechanics and
means the existence of a sufficiently high number of \emph{first integrals}.
To be more specific an Hamiltonian with Hamiltonian $H=H\left( p,q \right)$
system with $N$ degrees of freedom is said to be integrable if there exist
$N-1$ well defined\footnote{We say that a function is well-defined on the
phase space if it is \emph{analytic} and \emph{single-valued}.} 
functionally independent and Poisson-commuting first integrals 
\cite{Liouville1855,Whittaker}.
In the case of systems with infinitely many degrees of freedom, 
e.g. partial differential equations like the KdV equation, the 
existence of \emph{infinitely many conservation laws} is then required.
One of the most efficient way to find these infinitely many is the
existence a so-called \emph{Lax pair} \cite{Lax1968}.
A Lax pair is an associated \emph{overdetermined} linear problem
whose compatibility condition is guaranteed if and only if the desired
non-linear equation is satisfied.

Nowadays a purely algorithmic method to prove or disprove the existence
of a Lax pair is not available, so during the years many \emph{integrability
detectors} have been developed.
Integrability detectors are algorithmic procedure which are sufficient conditions
for integrability, or alternative definitions of integrability.
This means that integrability detectors can be used to prove the 
integrability of a given equation without the need of a Lax pair.
One of the fundamental integrability detectors, which works both at
continuous and discrete level, is the \emph{generalised symmetry approach}.
The generalised symmetry approach was mainly developed by the scientific 
school of A. B. Shabat in Ufa during the 80s and has obtained many important
result in the classification of partial differential equations 
\cite{AdlerShabatYamilov2000,HabibullinSokolovYamilov1995,MikhailovShabat1993,
MikhailovShabatSokolov1991,MikhailovShabatYamilov1987,Sokolov1988,SokolovShabat1984},
differential-difference equations 
\cite{Yamilov1983,Yamilov1984thesis,LeviYamilov1997,GarifullinYamilovLevi2016,GarifullinYamilovLevi2018}
and partial difference equations
\cite{LeviYamilov2009,LeviYamilov2011,GarifullinYamilov2015,GGY_autom}.
Another integrability detector is the algebraic entropy test.
The algebraic entropy test is specific to systems with discrete
degrees of freedom which can be put into a bi-rational form.
The basic idea, given a bi-rational map, which can be an ordinary difference
equation, a differential-difference equation or even a partial difference equation,
is to examine the growth of the degree of its iterates, and extract
a canonical quantity, which is an index of complexity of the map. 
This canonical quantity is what is called the algebraic entropy.
The idea of algebraic entropy as measure of the complexity of
the growth of bi-rational maps comes from the notion of complexity 
introduced by Arnol'd in \cite{Arnold1990} and was discussed for
the first time in relation of discrete systems by Veselov \cite{Veselov1992}.

In this paper we will compute the algebraic entropy of
two classes of first order, five-point differential-difference equation
which were classified recently in 
\cite{GarifullinYamilovLevi2016,GarifullinYamilovLevi2018}
through the generalised symmetry method.
We mention that these equations were used to produce and classify new 
examples of quad-equations in \cite{GGY_autom}.
These equations are autonomous fourth order differential-difference 
equations of the following form:
\begin{equation}
    \begin{aligned}
        \dv{u_{n}}{t} &= A\left( u_{n+1},u_{n},u_{n-1} \right)u_{n+2}
    +B\left( u_{n+1},u_{n},u_{n-1} \right)u_{n-2}
    \\
    &+C\left( u_{n+1},u_{n},u_{n-1} \right),
    \end{aligned}
    \label{eq:4thgen}
\end{equation}
where $u_{n}=u_{n}\left( t \right)$, 
is the dependent variable depending on $n\in\Z$ and on $t\in\R$.
Due to the similarity of equation \eqref{eq:4thgen} with the well-known
three-point Volterra equation
\begin{equation}
    \dv{u_{n}}{t} = u_{n}\left( u_{n+1}-u_{n-1} \right),
    \label{eq:2thgen}
\end{equation}
we will call equations of this form \emph{Volterra-like equations}.

\begin{remark}
    Throughout this paper we are going to consider only autonomous 
    equations of the form \eqref{eq:4thgen}.
    Therefore, we will make use of the short-hand notation $u_{n+k}=u_{k}$ 
    to simplify the formul\ae.
    \label{rem:notation}
\end{remark}

Explicitly, we are going to consider are the following equations
divided in six lists:
\setcounter{eqlistI}{0}
\setcounter{eqlistII}{0}
\begin{description}[%
  before={\setcounter{listcount}{0}},%
  ,font=\bfseries\stepcounter{listcount}List \thelistcount]
  \item Equations related to the double Volterra equation:
       \begin{align}
            \dv{u_0}{t}&=u_0(u_{2}-u_{-2}),
            \stepcounter{eqlistI}
            \tag{\theeqlistI}
            \label{Vol} 
            \\
            \dv{u_0}{t}&=u_0^2(u_{2}-u_{-2}),
            \stepcounter{eqlistI}
            \tag{\theeqlistI}
            \label{Vol0}
            \\
            \dv{u_0}{t}&=(u_0^2+u_0)(u_{2}-u_{-2}),
            \stepcounter{eqlistI}
            \tag{\theeqlistI}
            \label{Vol1}
            \\
            \dv{u_0}{t}&=(u_{2}+u_{1})(u_0+u_{-1})-(u_1+u_0)(u_{-1}+u_{-2}),
            \stepcounter{eqlistI}
            \tag{\theeqlistI}
            \label{Vol_mod}
            \\
            \dv{u_0}{t}&
            \begin{aligned}[t]
                &=(u_{2}-u_{1}+a)(u_0-u_{-1}+a)
                \\
                &+(u_1-u_0+a)(u_{-1}-u_{-2}+a)+b,
            \end{aligned}
            \stepcounter{eqlistI}
            \tag{\theeqlistI}
            \label{Vol_mod1}
            \\
            \dv{u_0}{t}&
            \begin{aligned}[t]
                &=u_{2}u_{1}u_0(u_0u_{-1}+1)\\
                &-(u_1u_0+1)u_0u_{-1}u_{-2}+u_0^2(u_{-1}-u_{1}),
            \end{aligned}
            \stepcounter{eqlistI}
            \tag{\theeqlistI}
            \label{Vol2}
            \\
           \dv{u_0}{t}&=u_{{0}} \left[u_1(u_2-u_0)+u_{-1}(u_0-u_{-2}) \right],
            \stepcounter{eqlistII}
            \tag{\theeqlistII}
            \label{eq1ii}
            \\
            \dv{u_0}{t}&=u_{{1}}{u_{{0}}}^{2}u_{-1} \left( u_{{2}}-u_{{-2}} \right).
            \stepcounter{eqlistII}
            \tag{\theeqlistII}
            \label{eq2ii}
       \end{align}
        Transformations $\tilde u_k=u_{2k}$ or $\tilde u_k=u_{2k+1}$ turn equations 
        \eqref{Vol}-\eqref{Vol1} into the well-known Volterra equation and its modifications 
    in their standard form.
       The other equations are related to the \emph{double
       Volterra equation} \eqref{Vol} through some autonomous 
        non-invertible non-point transformations. 
        We note that equation \eqref{eq2ii} was presented in 
        \cite{AdlerPostnikov2008}.
   \item Linearizable equations:
       \begin{align}
           \dv{u_0}{t}&
           \begin{aligned}[t]
               &=(T-a)
           \left[\frac{(u_1+au_0+b)(u_{-1}+au_{-2}+b)}{u_{0}+au_{-1}+b}+u_0+au_{-1}+b\right]
           \\
           &+cu_0+d,
           \end{aligned}
            \stepcounter{eqlistI}
            \tag{\theeqlistI}
            \label{Bur2}
            \\
            \dv{u_0}{t}&=\frac{u_2u_0}{u_1}+u_1-a^2\left(u_{-1}+\frac{u_0u_{-2}}{u_{-1}}\right)+cu_0.
            \stepcounter{eqlistI}
            \tag{\theeqlistI}
            \label{Bur}
       \end{align}
       In both equations $ a\neq0,$ in \eqref{Bur2} $(a+1)d=bc$, and $T$ is the translation
       operator $Tf_{n}=f_{n+1}$.

       Both equations of List \thelistcount\ are related to the linear equation:
       \begin{equation}
       \label{lin_eq} 
       \dv{u_0}{t}=u_2-a^2 u_{-2}+\frac{c}{2}u_0 
        \end{equation}
        through an autonomous non-invertible non-point transformations.
        We note that \eqref{Bur2} is linked to \eqref{lin_eq}
        with a transformation which is implicit in both directions,
        see \cite{GarifullinYamilovLevi2016} for more details.
    \item Equations related to a generalised symmetry of the  Volterra equation:
        \begin{align}
            \dv{u_0}{t}&
            \begin{aligned}[t]
                &=u_{{0}} 
                \left[u_1(u_2+u_1+u_0)-u_{-1}(u_0+u_{-1}+u_{-2}) \right]
                \\
                &+cu_{{0}} \left( u_{{1}}-u_{{-1}} \right),
            \end{aligned}
            \stepcounter{eqlistII}
            \tag{\theeqlistII}
            \label{Vol1s}
            \\
            \dv{u_0}{t} &
            \begin{aligned}[t]
                &= (u_{{0}}^2-a^2) 
            \left[(u_1^2-a^2)(u_2+u_0)-(u_{-1}^2-a^2)(u_0+u_{-2}) \right]
            \\
            &+c(u_{{0}}^2-a^2) \left( u_{{1}}-u_{{-1}} \right),
            \end{aligned}
            \stepcounter{eqlistII}
            \tag{\theeqlistII}
            \label{mVol2}
            \\
            \dv{u_0}{t}&=(u_1-u_0+a)(u_0-u_{-1}+a)(u_2-u_{-2}+4a+c)+b,
            \stepcounter{eqlistII}
            \tag{\theeqlistII}
            \label{Volz}
            \\
            \dv{u_0}{t}&=u_0[u_1(u_2-u_1+u_0)-u_{-1}(u_0-u_{-1}+u_{-2})],
            \stepcounter{eqlistII}
            \tag{\theeqlistII}
            \label{Vol_mod1s}
            \\
            \dv{u_0}{t}&=(u_{{0}}^2-a^2) \left[(u_1^2-a^2)(u_2-u_0)+(u_{-1}^2-a^2)(u_0-u_{-2}) \right],
            \stepcounter{eqlistII}
            \tag{\theeqlistII}
            \label{mVol3}
            \\
            \dv{u_0}{t}&=(u_1+u_0)(u_0+u_{-1})(u_2-u_{-2}).
            \stepcounter{eqlistII}
            \tag{\theeqlistII}
            \label{Vol_mod2}
        \end{align}
        These equations are related between themselves by some transformations, 
        for more details see \cite{GarifullinYamilovLevi2018}.
        Moreover equations (\ref{Vol1s},\ref{mVol2},\ref{Volz}) are the generalised symmetries of some known
        three-point autonomous 	differential-difference equations \cite{Yamilov2006}.
    \item  Equations of the relativistic Toda type:
        \begin{align}
            \dv{u_0}{t}&=(u_0-1)\left(\frac{u_2(u_1-1)u_0}{u_1}-\frac{u_0(u_{-1}-1)u_{-2}}{u_{-1}}-u_1+u_{-1}\right),
            \stepcounter{eqlistI}
            \tag{\theeqlistI}
        \label{our1}
        \\
        \dv{u_0}{t}&
        \begin{aligned}[t]
        &=\frac{u_2u_1^2u_0^2(u_0u_{-1}+1)}{u_1u_0+1}-\frac{(u_1u_0+1)u_0^2u_{-1}^2u_{-2}}{u_0u_{-1}+1}
        \\
        &-\frac{(u_1-u_{-1})(2u_1u_0u_{-1}+u_1+u_{-1})u_0^3}{(u_1u_0+1)(u_0u_{-1}+1)},
        \end{aligned}
            \stepcounter{eqlistI}
            \tag{\theeqlistI}
        \label{our2}
        \\
            \dv{u_0}{t}&=(u_1u_0-1)(u_0u_{-1}-1)(u_2-u_{-2}).
            \setcounter{eqlistII}{13}
            \tag{\theeqlistII} 
            \label{ourii}
        \end{align}
        Equation \eqref{ourii} was known 
        \cite{GarifullinYamilov2012,GarifullinMikhailovYamilov2014} to be
        is a relativistic Toda type equation.
        Since in \cite{GarifullinYamilovLevi2016} it was shown that the equations of List 
        \thelistcount\ are related through autonomous non-invertible non-point transformations,
        it was suggested that \eqref{our1} and \eqref{our2} should be of the same type.
        Finally, we note that equation \eqref{our1} appeared in \cite{Adler2016b} earlier than in 
        \cite{GarifullinYamilovLevi2016}.
        \setcounter{eqlistII}{8}
    \item Equations related to the Itoh-Narita-Bogoyavlensky (INB) equation:
        \begin{align}
            \label{INB}
            \dv{u_0}{t}&=u_0(u_2+u_1-u_{-1}-u_{-2}),
            \stepcounter{eqlistI}
            \tag{\theeqlistI}
            \\
            \stepcounter{eqlistI}
            \tag{\theeqlistI}
            \label{mod_INB}
            \dv{u_0}{t}&
            \begin{aligned}[t]
                &=(u_{2}-u_{1}+a)(u_0-u_{-1}+a)
                \\
                &+(u_1-u_0+a)(u_{-1}-u_{-2}+a)
                \\
                &+(u_1-u_0+a)(u_0-u_{-1}+a)+b,
            \end{aligned}
            \\
            \dv{u_0}{t} &= (u_0^2+au_0)(u_2u_1-u_{-1}u_{-2}),
            \stepcounter{eqlistI}
            \tag{\theeqlistI}
            \label{mikh1}
            \\
            \stepcounter{eqlistI}
            \tag{\theeqlistI}
            \label{eq1}
            \dv{u_0}{t} &= (u_1-u_0)(u_0-u_{-1})\left(\frac {u_2}{u_1}-\frac{u_{-2}}{u_{-1}}\right),
            \\
            \dv{u_0}{t}&=u_0(u_2u_1-u_{-1}u_{-2}),
            \stepcounter{eqlistII}
            \tag{\theeqlistII}
            \label{INB1}
            \\
            \dv{u_0}{t}&=(u_1-u_0+a)(u_0-u_{-1}+a)(u_2-u_1+u_{-1}-u_{-2}+2a)+b,
            \stepcounter{eqlistII}
            \tag{\theeqlistII}
            \label{INB3}
            \\
            \dv{u_0}{t}&=u_0(u_1u_0-a)(u_0u_{-1}-a)(u_2u_1-u_{-1}u_{-2}),
            \stepcounter{eqlistII}
            \tag{\theeqlistII}
            \label{MX2}
            \\
            \dv{u_0}{t}&=(u_1+u_0)(u_0+u_{-1})(u_2+u_1-u_{-1}-u_{-2}).
            \stepcounter{eqlistII}
            \tag{\theeqlistII}
            \label{INB2}
        \end{align}
        Equation \eqref{INB} is the well-known INB equation 
        \cite{Bogoyavlensky1988,Itoh1975,Narita1982}.
        Equations \eqref{mod_INB} with $a=0$ and \eqref{mikh1} with $a=0$ 
        are  simple modifications of the INB and were presented  in \cite{MikhailovXenitidis2013} 
        and \cite{Bogoyavlensky1991}, respectively.
        Equation \eqref{mikh1} with $a=1$ has been found in 
        \cite{Suris2003book}.
        Up to an obvious linear transformation, it is equation (17.6.24) 
        with $m=2$ in \cite{Suris2003book},.
        Equation \eqref{INB1} is a well-known modification of  INB equation \eqref{INB}, 
        found by Bogoyalavlesky himself \cite{Bogoyavlensky1988}. 
        Finally, equation \eqref{MX2} with $a=0$ was considered in \cite{AdlerPostnikov2008}. 
        All the equations in this list can be reduced to the INB equation
        using autonomous non-invertible non-point transformations.
        Moreover, equations \eqref{mod_INB},\eqref{eq1} and \eqref{INB1} are related through
        non-invertible transformations to the equation:
        \begin{align}
            \dv{u_{0}}{t} &= \left( u_{2}-u_{0} \right)\left( u_{1}-u_{-1} \right)
                \left( u_{0}-u_{-2} \right).
            \label{eq311}
        \end{align}
        For this reason, as it was done in \cite{GGY_autom}, we will consider equation
        \eqref{eq311}, as independent.
        We note that equation \eqref{eq311} and its relationship with equation 
        \eqref{mod_INB},\eqref{eq1} and \eqref{INB1} were first discussed 
        in \cite{GarifullinYamilovLevi2016_non_inv}.
    \item Other equations:
        \begin{align}
            \dv{u_0}{t} &= u_0^2(u_2u_1-u_{-1}u_{-2})-u_0(u_1-u_{-1}),
            \stepcounter{eqlistI}
            \tag{\theeqlistI}
            \label{seva}
            \\
            \dv{u_0}{t} &
            \begin{aligned}[t]
                &= (u_0+1)\times
                \\
                &\phantom{\times}\left[\frac{u_2u_0(u_{1}+1)^2}{u_1}-\frac{u_{-2}u_0(u_{-1}+1)^2}{u_{-1}}
            +(1+2u_0)(u_1-u_{-1})\right],
            \end{aligned}
            \stepcounter{eqlistI}
            \tag{\theeqlistI}
            \label{rat}
            \\
            \dv{u_0}{t} &= (u_0^2+1)\left(u_2\sqrt{u_1^2+1}-u_{-2}\sqrt{u_{-1}^2+1}\right),
            \stepcounter{eqlistI}
            \tag{\theeqlistI}
            \label{sroot}
            \\
            \dv{u_0}{t}&=u_1u_0^3u_{-1}(u_2u_1-u_{-1}u_{-2})-u_0^2(u_1-u_{-1}).
            \setcounter{eqlistII}{14}
            \tag{\theeqlistII}
            \label{SK2}
        \end{align}
        Equation \eqref{seva} has been found in \cite{TsujimotoHirota1996} 
        and it is called the discrete Sawada-Kotera equation 
        \cite{Adler2011,TsujimotoHirota1996}. 
        Equation \eqref{SK2} is a simple modification of the discrete
        Sawada-Kotera equation \eqref{seva}. 
                Equation \eqref{rat} has been found in \cite{Adler2016b} and 
        is related to \eqref{seva}.
        On the other hand equation \eqref{sroot} has been found as 
        a result of the classification in \cite{GarifullinYamilovLevi2016} 
        and seems to be a new equation. 
        It was shown in \cite{GarifullinYamilov2017} that equation 
        \eqref{sroot} is a discrete analogue of the Kaup-Kupershmidt
        equation.
        Then we will refer to equation \eqref{sroot} as the discrete
        Kaup-Kupershmidt equation.
        No transformation into known equations 
        of equation \eqref{sroot} is known.
\end{description}

In this paper we will show that all the bi-rational equations of Lists 1--6 possess
quadratic (linear) growth, and hence are integrable (linearizable) according
to the algebraic entropy test.
We note that the only non-bi-rational equations equation in Classes I and
II is the discrete Kaup-Kupershmidt equation \eqref{sroot} which contains
square root terms.
In section \ref{sec:theory} we will give some details on how algebraic
entropy is computed, then in
section \ref{sec:results} we will show the results for the equations of
Lists 1--6 except the discrete Kaup-Kupershmidt equation \eqref{sroot}.
In section \ref{sec:disc} we will discuss the results obtained in sections
\ref{sec:results} 
in the framework of the existing literature
and we will give an outlook on future research in the field.

Before going on we would like to present a new \emph{rational} form
of the discrete Kaup-Kupershmidt equation \eqref{sroot}.
That is we have the following proposition:

\begin{proposition}
    There exists a point transformation which brings the discrete
    Kaup-Kupershmidt equation \eqref{sroot} into the following rational form:
    \begin{equation}
        \dv{v_{0}}{t}
        =\left( 1+v_{0}^{2} \right)
        \left[ \frac{1+v_{1}^{2}}{1-v_{1}^{2}}\frac{v_{2}}{1-v_{2}^{2}}- 
        \frac{1+v_{-1}^{2}}{1-v_{-1}^{2}}\frac{v_{-2}}{1-v_{-2}^{2}}\right].
        \label{srootrat}
    \end{equation}
    \label{prop:ratkk}
\end{proposition}
\begin{proof}
    We start with the substitution:
    \begin{equation}
        u_{n} = \sinh\left( \varphi_{n} \right),
        \label{eq:hyp}
    \end{equation}
    which brings the discrete Kaup-Kupershmidt equation in hyperbolic
    form\footnote{D. Levi, private communication.}:
    \begin{equation}
        \dot{\varphi}_{0} = \cosh\left( \varphi_{n} \right)
        \left[ \cosh\left( \varphi_{n+1} \right)\sinh\left( \varphi_{n+2} \right)-
        \cosh\left( \varphi_{n-1}\right)\sinh\left( \varphi_{n-2} \right) \right].
        \label{sroothyp}
    \end{equation}
    Using the hyperbolic identities:
    \begin{equation}
        \sinh\alpha = \frac{2 \tanh\left( \alpha/2 \right)}{1-\tanh^{2}\left( \alpha/2 \right)},
        \quad
        \cosh\alpha = \frac{1+ \tanh^{2}\left( \alpha/2 \right)}{1-\tanh^{2}\left( \alpha/2 \right)},
        \label{eq:hypid}
    \end{equation}
    and putting
    \begin{equation}
        \tanh\left( \frac{\varphi_{n}}{2} \right) = v_{n}
        \label{eq:convhyp}
    \end{equation}
    equation \eqref{srootrat} follows.
\end{proof}

\begin{remark}
    We note that under the scaling:
    \begin{equation}
        u_{n}(t) = \imath\left[ \sqrt{2}+\frac{\sqrt{2}}{8}\epsilon^2 
        U\left(\tau-\frac{2}{135}\epsilon^5 t, x+\frac{4}{9}\epsilon t\right)\right],
        \,
        x = n\varepsilon,
        \label{eq:clim}
    \end{equation}
    equation \eqref{srootrat} admits the Kaup-Kaupershmit equation as continuum limit:
    \begin{equation}
        U_{\tau} = U_{xxxxx} + 5 U U_{xxx} + \frac{25}{2} U_{x}U_{xx}+5 U^{2}U_{x},
        \label{eq:kk}
    \end{equation}
    just as the original \eqref{sroot} equation.
    \label{rem:cl}
\end{remark}


\section{Algebraic entropy}
\label{sec:theory}

Heuristically integrability deals with the regularity of the
solutions of a given system.
In this sense a simple characterisation of chaotic behaviour
is when two arbitrarily near initial values give rise to solutions
diverging at infinity.
For recurrence relations, i.e. equations where the solution is
given by iteration of a formula, we could just try to compute
the iteration to extract information about integrability,
even if we cannot solve the equation explicitly. 
However it is usually impossible to calculate explicitly these 
iterates by hand or even with any state-of-the-art
formal calculus software, simply because the expressions one should manipulate
are rational fractions of increasing degree of the various initial conditions. 
The complexity and size of the calculation make it impossible to  calculate the iterates.

It was nevertheless observed that ``integrable'' maps are not as
complex as generic ones. This was done primarily experimentally, 
by an accumulation of examples, and later by the elaboration of the concept of 
\emph{algebraic entropy} for difference equations
\cite{BellonViallet1999,FalquiViallet1993,Veselov1992,Diller1996,Russakovskii1997}.
In \cite{Tremblay2001,Viallet2006} the method was developed in the case
of quad equations and then used as a classifying tool \cite{HietarintaViallet2007}.
Finally in \cite{DemskoyViallet2012} the same concept
was introduced for differential-difference equation and later \cite{viallet2014}
to the very similar case of differential-delay equations.
For a more complete discussion of the method in the context of the so-called
integrability indicators we refer to 
\cite{GrammaticosHalburdRamaniViallet2009,GubbiottiASIDE16}.


As we stated in the introduction algebraic entropy is a measure of
the growth of bi-rational maps.
The most natural space for considering bi-rational maps is the projective
space over a closed field rather than in the affine space one.
We then transforms a recurrence relation into a polynomial map in 
the homogeneous coordinates of the
proper projective space over some closed field:
\begin{equation}
    \varphi \colon x_{i} \mapsto \varphi_{i}\left( x_{k} \right),
    \label{eq:projmap}
\end{equation}
with $x_{i},x_{k}\in\IN$ where $\IN$ is the space of the initial conditions.
The recurrence is then obtained by iterating the polynomial map $\varphi$.
The map $\varphi$ has to be \emph{bi-rational} in the sense that it
has to possesses an inverse map which is again a rational map.

The space of the initial condition depends on which type of recurrence relation
we are considering. 
In the case of differential-difference equation of the discrete $k-k'$-th order
and of the $p$-th continuous order:
\begin{equation}
    u_{n+k} = f_{n}\left( \left\{\dv[i]{u_{n+k-1}}{t},\ldots,
    \dv[i]{u_{n+k'+1}}{t}\right\}_{i=0}^{p};u_{n+k'}\right),
    \quad 
    k',k,n\in\Z,\, k'<k
    \label{eq:ddele}
\end{equation}
the space of initial conditions is infinite dimensional. 
Indeed, in the case the order of the equation is $k-k'$, 
we need the initial value of $k-k'$-tuple as a function of the
parameter $t$, but also the value of \emph{all its derivatives}:
\begin{equation}
    \IN = \left\{\dv[i]{u_{k-1}}{t},
        \dv[i]{u_{k-2}}{t},
        \dots,
        \dv[i]{u_{k'}}{t},
        \right\}_{i\in\N_{0}}.
    \label{eq:inidde}
\end{equation}
We need all the derivatives of $u_{i}\left( t \right)$
and not just the first $p$ because at every iteration the order
of the equation is raised by $p$.
Therefore, to describe infinitely many iterations we need
infinitely many derivatives.
To obtain the map one just need to pass to homogeneous 
coordinates in the equation and in \eqref{eq:inidde}.

\begin{remark}
    We remark that if we restrict to compute a finite number of iterates
    of a differential-difference equation of discrete $k-k'$-th order
    and of the $p$-th continuous order \eqref{eq:ddele} then the space 
    of initial conditions is finite dimensional.
    Indeed, let us assume that we wish to compute the $N$th iterate of a 
    differential-difference equation \eqref{eq:ddele}, then at most we will
    need the derivatives of order $N\left( p+1 \right)$.
    That is, we need to consider the following restricted space of initial
    conditions:
    \begin{equation}
        \IN^{(N)} = \left\{\dv[i]{u_{k-1}}{t},
            \dv[i]{u_{k-2}}{t},
        \dots, \dv[i]{u_{k'}}{t},\right\}_{i=0}^{N\left( p+1 \right)}.
        \label{eq:iniddeN}
    \end{equation}
    \label{rem:raiseorder}
\end{remark}

If we factor out any common polynomial factors we can say that 
the degree with respect to the initial conditions is well defined.
We can therefore form the sequence of degrees
of the iterates of the map $\varphi$ and call it $d_{N}=\deg \varphi^{N}$:
\begin{equation}
    1,d_{1},d_{2},d_{3},d_{4},d_{5},\dots,d_{N},\dots.
    \label{eq:onedimdegrees}
\end{equation}
The degree of the bi-rational projective map $\varphi$ has to
be understood as the \emph{maximum of the total polynomial degree in the initial
conditions} $\IN$ of the entries of $\varphi$.
The same definition in the affine case just translates to the \emph{maximum
of the degree of the numerator and of the denominator} of
the $N$th iterate in terms of the affine initial conditions.
Degrees in the projective and in the affine setting can be different,
but the global behaviour will be the same due to the properties of homogenization
and de-homogenization.

The sequence of degree \eqref{eq:onedimdegrees} is fixed in a given 
system of coordinates, but it is not invariant with respect 
to changes of coordinates.
Therefore we need to introduce a canonical measure of the growth.
It turn out that a good definition is the following one:
consider the following number
\begin{equation}
    \eta_{\varphi} = \lim_{N\to\infty} \frac{1}{N}\log d_{N},
    \label{eq:algent}
\end{equation}
called the \emph{algebraic entropy} of the map $\varphi$
When no confusion is possible about the map $\varphi$ we will usually
omit the subscript $\varphi$ in \eqref{eq:algent}.

Algebraic entropy for bi-rational maps has the following properties 
\cite{BellonViallet1999,GubbiottiASIDE16,GrammaticosHalburdRamaniViallet2009}:
\begin{enumerate}
    \item The algebraic entropy as given by \eqref{eq:algent} always exists.
    \item The algebraic entropy has the following upper bound:
        \begin{equation}
            \eta_{\varphi} \leq \deg \phi.
            \label{eq:algentineq}
        \end{equation}
    \item If $\eta_{\varphi}=0$, i.e. the algebraic entropy is zero, then
        \begin{equation}
            d_{N} \sim N^{\nu}, \quad \text{with $\nu \in \N_{0}$, as $N\to\infty$.}
            \label{eq:dkasympt}
        \end{equation}
    \item The algebraic entropy is a \emph{bi-rational invariant of bi-rational maps}.
        That is, if two bi-rational maps $\varphi$ and $\psi$ are conjugated
        by a bi-rational map $\chi$,
        \begin{equation}
            \varphi = \chi \circ \psi \circ \chi^{-1}
            \label{eq:conj}
        \end{equation}
        then:
        \begin{equation}
            \eta_{\varphi} = \eta_{\psi}.
            \label{eq:invariance}
        \end{equation}
        \label{prop:invariance}
\end{enumerate}
Property \ref{prop:invariance} tell us that the algebraic entropy 
is a canonical measure of growth for bi-rational maps.

We will then have the following classification of
equations according to their Algebraic Entropy
\cite{HietarintaViallet2007}:
\begin{description}
    \item[Linear growth] The equation is linearizable.
    \item[Polynomial growth] The equation is integrable.
    \item[Exponential growth] The equation is chaotic.
\end{description}

In our the following sections we will be dealing with 
differential-difference equations of first continuous order
and fourth discrete order of the particular form:
\begin{equation}
    u_{n+2} = f\left( u_{n+1},u_{n},u_{n-1},u_{n-2},\dv{u_{n}}{t}\right).
    \quad n\in\Z,
    \label{eq:ddelespec}
\end{equation}

To practically compute the algebraic entropy we introduce
some technical methods to reduce the computational complexity
\cite{GubbiottiASIDE16,GubbiottiPhD2017}.
First, we fix the desired number of iterations to be some fixed $N\in\N$.
Following remark \ref{rem:raiseorder} this means that we need only 
finitely many initial conditions given by \eqref{eq:iniddeN}. 
Then we assume that the space of initial conditions is 
\emph{linearly parametrised} in the appropriate projective space, 
i.e. in inhomogenous coordinates it has the following form:
\begin{equation}
    u_{i} = \frac{\alpha_{i} t + \beta_{i}}{\alpha_{0}t+\beta_{0}},
    \quad
    u_{i} \in \IN^{(N)}.
    \label{eq:linpar}
\end{equation}
We will assume that the parameter $t$ is the same which describes the 
``time'' evolution of the problem.
To simplify the problem we choose all the parameters involved in the 
equations to be integers.
Moreover, to avoid accidental factorisations which may alter the
results we choose these integers to be \emph{prime numbers}.
A final simplification to speed up the computations is given
by considering the factorisation of the iterates in some finite field
$\K_{r}$, with $r$ prime number.

\begin{remark}
    Several equations in Lists 1--6, e.g. \eqref{Vol_mod1}
    or \eqref{Volz}, depend on some parameters.
    Depending on the value of the parameters their integrability properties can be, 
    in principle, different.
    As it was done in \cite{GGY_autom}, in order to avoid ambiguities, 
    we use some simple autonomous transformations to fix the values of some parameters.
    The remaining free parameters are then treated as free coefficients
    and then fixed to integers following the above discussion.
    We will describe these subcases when needed in the next section.
    \label{rem:params}
\end{remark}

Using the rules above we are able to avoid accidental cancellations
and produce a finite sequence of degrees:
\begin{equation}
    1,d_{1},d_{2},d_{3},d_{4},d_{5},\dots,d_{N}.
    \label{eq:onedimdegreesfin}
\end{equation}
To extract the asymptotic behaviour from the finite sequence \eqref{eq:onedimdegreesfin}:
we compute its generating function, i.e. a function $g=g\left( s \right)$ such
that the coefficients of its Taylor series 
\begin{equation}
    g(z) = \sum_{l=0}^{\infty} d_{l} z^{l}
    \label{eq:genfunc}
\end{equation}
up to order $N$ coincides with the finite sequence \eqref{eq:onedimdegreesfin}.
Assuming that such generating function is rational it can be 
computed using the method of Pad\'e approximants
\cite{Pade1892,BakerGraves1996}. 

This generating function is \emph{predictive} tool. 
Indeed one can readily compute the successive terms
in the Taylor expansion for \eqref{eq:genfunc} and confront them with
the degrees calculated with the iterations. 
This means that the assumption that the value of the algebraic entropy given 
by the approximate method is in fact very strong and very unlikely the 
real value will differ from it. 

Having a rational generating function will also yield the value
of the Algebraic Entropy from the modulus of the
smallest pole of the generating function:
\begin{equation}
    \eta_{\varphi} = 
    \log \min\left\{|z|\in\R^{+} \,\,\middle|\,\, \frac{1}{g\left( z \right)}=0\right\}.
    \label{eq:algentgenfunc}
\end{equation}
From the generating function one can also find an asymptotic fit for 
the degrees \eqref{eq:onedimdegreesfin}. 
This can be done by using the $\mathcal{Z}$-transform \cite{Elaydi2005,Jury1964}.
Indeed, from the definition of $\mathcal{Z}$-transform it can be readily proved 
that:
\begin{equation}
    d_{l} = \mathcal{Z}\left[g\left( \frac{1}{\zeta} \right)\right]_{l},
    \label{eq:ztransform}
\end{equation}
where $\mathcal{Z}[f(\zeta)]_{l}$ is the $\mathcal{Z}$-transform of
the function $f(\zeta)$.

\begin{remark}
    We note that the  general asymptotic behaviour of the
    sequence $\left\{ d_{l} \right\}_{l\in\N_{0}}$ can be obtained even without computing the 
    $\mathcal{Z}$-transform.
    Indeed, let us assume that the given generating $g$ function
    has radius of convergence $\rho>0$. 
    Then, let us assume that the generating function can be written in the following way:
    \begin{equation}
        g = A\left( z \right) + B\left( z \right)\left( 1-\frac{z}{\rho} \right)^{-\beta},
        \label{eq:gasy}
    \end{equation}
    where $A$ and $B$ are analytic functions for $\abs{z}<r$ such that
    $B(\rho)\neq 0$.
    Then we have the following estimate:
    \begin{equation}
        d_{N} \sim \frac{B\left( \rho \right)}{\Gamma\left( \beta \right)}
        N^{\beta-1}\rho^{-N},
        \quad
        N \to \infty
        \label{eq:dasy}
    \end{equation}
    where $\Gamma(z)$ is the Euler Gamma function.
    When the radius of convergence is one, i.e. when the given equation
    is integrable, we have the simpler estimate
    \begin{equation}
        d_{N} \sim N^{\beta-1},
        \quad
        N \to \infty.
        \label{eq:dasyint}
    \end{equation}
    \label{rem:asydeg}
\end{remark}

\section{Results}
\label{sec:results}

In this section we describe the results of the procedure outlined in section 
\ref{sec:theory} for the differential-difference equations of Lists 1--6.
Specifically, as described in Remark \ref{rem:params}, we will underline
the particular cases in which the parametric equations can be divided.
We notice that certain equations are \emph{symmetric} under the involution
\begin{equation}
    u_{n} \to \tilde{u}_{n} = u_{-n}.
    \label{eq:inv}
\end{equation}
This implies that the recurrence defined by solving the
equation with respect to $u_{2}$ and $u_{-2}$ is the same.
For equations satisfying this property the growth of
the degree of the iterates can be computed just in one direction,
as the growth in the other direction will be the same.
Computations are performed using the \texttt{python} program for
differential-difference equations presented in \cite{GubbiottiPhD2017}.
We remark that this program was already employed to discuss
the integrability of some three-point differential-difference
equations in \cite{GSL_symmetries}.

\subsection{List 1}

\subsubsection{Equation \eqref{Vol}:}
Equation \eqref{Vol} is symmetric and has the following growth of degrees:
\begin{equation}
    1, 2, 2, 4, 4, 7, 7, 11, 11, 16, 16, 22, 22\dots.
    \label{eq:growthVol}
\end{equation}
The generating function corresponding to the growth 
\eqref{eq:growthVol} is:
\begin{equation}
    g(z) = -\frac{z^4 - 2 z^2 + z + 1}{(z - 1)^3 (z + 1)^2}.
    \label{eq:gfVol}
\end{equation}
All the poles of $g$ lie on the unit circle, so that the entropy
is zero.
Moreover, due to the presence of the factor $\left( z-1 \right)^{3}$
following remark \ref{rem:asydeg} we have that equation
\eqref{Vol} has quadratic growth.

\subsubsection{Equation \eqref{Vol0}:}
Equation \eqref{Vol0} is symmetric and has the following growth of degrees:
\begin{equation}
    1, 3, 3, 7, 7, 13, 13, 21, 21, 31, 31, 43, 43\dots.
    \label{eq:growthVol0}
\end{equation}
The generating function corresponding to the growth 
\eqref{eq:growthVol0} is:
\begin{equation}
    g(z) = -\frac{z^4 - 2 z^2 + 2 z + 1}{(z - 1)^3 (z + 1)^2}.
    \label{eq:gfVol0}
\end{equation}
All the poles of $g$ lie on the unit circle, so that the entropy
is zero.
Moreover, due to the presence of the factor $\left( z-1 \right)^{3}$
following remark \ref{rem:asydeg} we have that equation
\eqref{Vol0} has quadratic growth.

\subsubsection{Equation \eqref{Vol1}:}
Equation \eqref{Vol1} is symmetric and has the same growth of degrees
as equation \eqref{Vol0}.
Therefore we have that equation \eqref{Vol1} has zero entropy
and quadratic growth.

\subsubsection{Equation \eqref{Vol_mod}:}
Equation \eqref{Vol_mod} is symmetric and has the following growth of degrees:
\begin{equation}
    1, 2, 3, 5, 6, 8, 10, 14, 16, 20, 23, 29, 32, 38, 42, 50, 54\dots.
    \label{eq:growthVol_mod}
\end{equation}
The generating function corresponding to the growth 
\eqref{eq:growthVol_mod} is:
\begin{equation}
    g(z) = -\frac{z^7 + z^6 - z^5 - z^4 + z^3 + z + 1}{(z - 1)^3 (z + 1)^2 (z^2 + 1)}.
    \label{eq:gfVol_mod}
\end{equation}
All the poles of $g$ lie on the unit circle, so that the entropy
is zero.
Moreover, due to the presence of the factor $\left( z-1 \right)^{3}$
following remark \ref{rem:asydeg} we have that equation
\eqref{Vol_mod} has quadratic growth.

\subsubsection{Equation \eqref{Vol_mod1}:}
Equation \eqref{Vol_mod1} depends on the parameter $a$.
Using a simple scaling if $a\neq0$ it is possible to set $a=1$. 
For this reason we can consider the two cases $a=1$ and $a=0$.
If $a=1$ equation \eqref{Vol_mod1} is asymmetric, but it has
the following growth of degrees in both directions:
\begin{equation}
    \begin{aligned}
        &1, 2, 3, 5, 6, 8, 10, 14, 16, 20, 23, 29, 32, 38, 
        \\
        &\quad42,50, 54, 62, 67, 77, 82, 92, 98, 110, 116\dots.
    \end{aligned}
    \label{eq:growthVol_mod1}
\end{equation}
The generating function corresponding to the growth 
\eqref{eq:growthVol_mod1} is:
\begin{equation}
    g(z) = -\frac{z^7 + z^6 - z^5 - z^4 + z^3 + z + 1}{(z - 1)^3 (z + 1)^2 (z^2 + 1)}.
    \label{eq:gfVol_mod1}
\end{equation}
If $a=0$ equation \eqref{Vol_mod1} is symmetric, but its growth of
degrees is still given by the sequence \eqref{eq:growthVol_mod1}
and fitted by the generating function \eqref{eq:gfVol_mod1}.
Therefore in both cases the entropy is zero since all the poles of 
$g$ lie on the unit circle.
Moreover, due to the presence of the factor $\left( z-1 \right)^{3}$
following remark \ref{rem:asydeg} we have that equation
\eqref{Vol_mod1} has quadratic growth for all values of $a$.

%
%

\subsubsection{Equation \eqref{Vol2}:}
Equation \eqref{Vol2} is symmetric and has the following growth of degrees:
\begin{equation}
    1, 5, 8, 14, 19, 28, 35, 47, 56, 71, 82, 100, 113\dots.
    \label{eq:growthVol2}
\end{equation}
The generating function corresponding to the growth 
\eqref{eq:growthVol2} is:
\begin{equation}
    g(z) = -\frac{z^5 - 2 z^3 + z^2 + 4 z + 1}{(z - 1)^3 (z + 1)^2}.
    \label{eq:gfVol2}
\end{equation}
All the poles of $g$ lie on the unit circle, so that the entropy
is zero.
Moreover, due to the presence of the factor $\left( z-1 \right)^{3}$
following remark \ref{rem:asydeg} we have that equation
\eqref{Vol2} has quadratic growth.

\subsubsection{Equation \eqref{eq1ii}:}
Equation \eqref{eq1ii} is symmetric and has the following growth of degrees:
\begin{equation}
    1, 3, 4, 6, 8, 12, 15, 19, 24, 29, 34, 40, 47, 54, 61, 69, 78, 87, 96, 106, 117\dots.
    \label{eq:growtheq1ii}
\end{equation}
The generating function corresponding to the growth 
\eqref{eq:growtheq1ii} is:
\begin{equation}
    g(z) = -\frac{z^9 - 3 z^8 + 4 z^7 - 4 z^6 + 4 z^5 - 3 z^4 + 2 z^3 - z^2 + 1}{(z - 1)^3 (z^2 + 1)}.
    \label{eq:gfeq1ii}
\end{equation}
All the poles of $g$ lie on the unit circle, so that the entropy
is zero.
Moreover, due to the presence of the factor $\left( z-1 \right)^{3}$
following remark \ref{rem:asydeg} we have that equation
\eqref{eq1ii} has quadratic growth.

\subsubsection{Equation \eqref{eq2ii}:}
Equation \eqref{eq2ii} is symmetric and has the following growth of degrees:
\begin{equation}
    \begin{aligned}
        &1, 5, 7, 13, 18, 27, 34, 45, 54, 69, 80, 97, 110,
        \\
        &\quad 131,146, 169, 186, 213, 232, 261, 282\dots.
    \end{aligned}
    \label{eq:growtheq2ii}
\end{equation}
The generating function corresponding to the growth 
\eqref{eq:growtheq2ii} is:
\begin{equation}
    g(z) = -\frac{z^9 - z^8 + z^6 - z^5 + 2 z^4 + 2 z^3 + z^2 + 4 z + 1}{(z - 1)^3 (z + 1)^2 (z^2 + 1)}.
    \label{eq:gfeq2ii}
\end{equation}
All the poles of $g$ lie on the unit circle, so that the entropy
is zero.
Moreover, due to the presence of the factor $\left( z-1 \right)^{3}$
following remark \ref{rem:asydeg} we have that equation
\eqref{eq2ii} has quadratic growth.

\subsection{List 2}

\subsubsection{Equation \eqref{Bur2}:}

Equation \eqref{Bur2} depends on four parameters $a$, $b$, $c$ and $d$
linked among themselves by the condition $\left( a+1 \right)d=bc$.
Using a linear transformation $u_{n,m}\to\alpha u_{n,m}+\beta$
we need to consider only three different cases: 
\begin{enumerate}
    \item $ a\neq0,a\neq -1,b=0,d=0$, 
    \item $a=-1,b=1,c=0$,
    \item $a=-1,b=0$.
\end{enumerate}
Recall that $a\neq0$ in all cases.
See \cite{GGY_autom} for more details.
In all the three cases equation \eqref{Bur2} is asymmetric.
However, it has the same growth of degrees in both directions 
and in all the three cases:
\begin{equation}
    1, 4, 7, 10, 13, 16, 19, 22, 25, 28, 31, 34, 37\dots.
    \label{eq:growthBur2}
\end{equation}
The generating function corresponding to the growth 
\eqref{eq:growthBur2} is:
\begin{equation}
    g(z) = -\frac{z^7 + z^6 - z^5 - z^4 + z^3 + z + 1}{(z - 1)^3 (z + 1)^2 (z^2 + 1)}.
    \label{eq:gfBur2}
\end{equation}
In all cases the entropy is zero since all the poles of 
$g$ lie on the unit circle.
Moreover, due to the presence of the factor $\left( z-1 \right)^{2}$
following remark \ref{rem:asydeg} we have that equation
\eqref{Bur2} has linear growth for all values of the parameters.

%

%
%

\subsubsection{Equation \eqref{Bur}:}
Equation \eqref{Bur} is not symmetric, but in both directions has the 
following growth of degrees:
\begin{equation}
    1, 4, 6, 9, 11, 14, 16, 19, 21, 24, 26, 29, 31\dots.
    \label{eq:growthBur}
\end{equation}
The generating function corresponding to the growth 
\eqref{eq:growthBur} is:
\begin{equation}
    g(z) = \frac{z^2 + 3 z + 1}{(z - 1)^2 (z + 1)}.
    \label{eq:gfBur}
\end{equation}
All the poles of $g$ lie on the unit circle, so that the entropy
is zero.
Moreover, due to the presence of the factor $\left( z-1 \right)^{2}$
following remark \ref{rem:asydeg} we have that equation
\eqref{Bur} has linear growth.

\subsection{List 3}
\subsubsection{Equation \eqref{Vol1s}}
Equation \eqref{Vol1s} is symmetric and has the following growth of degrees:
\begin{equation}
    1, 3, 6, 10, 16, 22, 29, 37, 46, 56, 67, 79, 92\dots.
    \label{eq:growthVol1s}
\end{equation}
The generating function corresponding to the growth 
\eqref{eq:growthVol1s} is:
\begin{equation}
    g(z) = -\frac{z^6 - 2 z^5 + z^4 + 1}{(z - 1)^3}.
    \label{eq:gfVol1s}
\end{equation}
All the poles of $g$ lie on the unit circle, so that the entropy
is zero.
Moreover, due to the presence of the factor $\left( z-1 \right)^{3}$
following remark \ref{rem:asydeg} we have that equation
\eqref{Vol1s} has quadratic growth.

\subsubsection{Equation \eqref{mVol2}}
Equation \eqref{mVol2} depends on the parameter $a$.
Using a simple scaling if $a\neq0$ it is possible to set $a=1$. 
For this reason we can consider the two cases $a=1$ and $a=0$.
Equation \eqref{mVol2} is symmetric for both $a=1$ and $a=0$.
Moreover, in both cases it has the following growth of degrees:
\begin{equation}
    1, 5, 11, 21, 31, 43, 57, 73, 91, 111, 133, 157, 183\dots.
    \label{eq:growthmVol2}
\end{equation}
The generating function corresponding to the growth 
\eqref{eq:growthmVol2} is:
\begin{equation}
    g(z) = -\frac{2 z^5 - 4 z^4 + 2 z^3 - z^2 + 2 z + 1}{z - 1)^3}.
    \label{eq:gfmVol2}
\end{equation}
Therefore in both cases the entropy is zero since all the poles of 
$g$ lie on the unit circle.
Moreover, due to the presence of the factor $\left( z-1 \right)^{3}$
following remark \ref{rem:asydeg} we have that equation
\eqref{mVol2} has quadratic growth for all values of $a$.

%
%

\subsubsection{Equation \eqref{Volz}:}
Equation \eqref{Volz} is not symmetric, but it has the same growth of degrees 
in both directions:
\begin{equation}
    1, 3, 4, 7, 10, 15, 19, 25, 31, 39, 46, 55, 64\dots.
    \label{eq:growthVolz}
\end{equation}
The generating function corresponding to the growth 
\eqref{eq:growthVolz} is:
\begin{equation}
    g(z) = -\frac{z^5 - z^4 + 2 z^3 - z^2 + z + 1}{(z - 1)^3 (z + 1) (z^2 + 1)}.
    \label{eq:gfVolz}
\end{equation}
All the poles of $g$ lie on the unit circle, so that the entropy
is zero.
Moreover, due to the presence of the factor $\left( z-1 \right)^{3}$
following remark \ref{rem:asydeg} we have that equation
\eqref{Volz} has quadratic growth.

\subsubsection{Equation \eqref{Vol_mod1s}:}
Equation \eqref{Vol_mod1s} is symmetric and has the same growth of degrees
as equation \eqref{Vol1s}.
Therefore we have that equation \eqref{Vol_mod1s} has zero entropy
and quadratic growth.

\subsubsection{Equation \eqref{mVol3}:}
Equation \eqref{mVol3} depends on the parameter $a$.
Using a simple scaling if $a\neq0$ it is possible to set $a=1$. 
For this reason we can consider the two cases $a=1$ and $a=0$.
Equation \eqref{mVol3} is symmetric for both $a=1$ and $a=0$.
However, in both cases equation \eqref{mVol3} has the same growth of 
degrees as equation \eqref{mVol2}.
Therefore we have that equation \eqref{mVol3} has zero entropy
and quadratic growth for all values of $a$.
%
%
%
%
%

\subsubsection{Equation \eqref{Vol_mod2}:}
Equation \eqref{Vol_mod2} is symmetric and has the following growth of degrees:
\begin{equation}
    1, 3, 4, 7, 10, 15, 19, 25, 31, 39, 46, 55, 64\dots.
    \label{eq:growthVol_mod2s}
\end{equation}
The generating function corresponding to the growth 
\eqref{eq:growthVol_mod2s} is:
\begin{equation}
    g(z) = -\frac{z^5 - z^4 + 2 z^3 - z^2 + z + 1}{(z - 1)^3 (z + 1) (z^2 + 1)}.
    \label{eq:gfVol_mod2s}
\end{equation}
All the poles of $g$ lie on the unit circle, so that the entropy
is zero.
Moreover, due to the presence of the factor $\left( z-1 \right)^{3}$
following remark \ref{rem:asydeg} we have that equation
\eqref{Vol_mod2} has quadratic growth.

\subsection{List 4}

\subsubsection{Equation \eqref{our1}}
Equation \eqref{our1} is symmetric and has the following growth of degrees:
\begin{equation}
    1, 5, 10, 16, 26, 38, 51, 65, 82, 102, 123, 145, 170, 198, 227, 257\dots.
    \label{eq:growthour1}
\end{equation}
The generating function corresponding to the growth 
\eqref{eq:growthour1} is:
\begin{equation}
    g(z) = -\frac{z^8 - 2 z^7 + 2 z^6 - 2 z^5 + z^4 + 2 z^3 - z^2 + 2 z + 1}{(z - 1)^3 (z^2 + 1)}.
    \label{eq:gfour1}
\end{equation}
All the poles of $g$ lie on the unit circle, so that the entropy
is zero.
Moreover, due to the presence of the factor $\left( z-1 \right)^{3}$
following remark \ref{rem:asydeg} we have that equation
\eqref{our1} has quadratic growth.

\subsubsection{Equation \eqref{our2}}
Equation \eqref{our2} is symmetric and has the following growth of degrees:
\begin{equation}
    \begin{aligned}
        &1, 9, 19, 37, 55, 75, 101, 129, 163, 199,237,
        \\
        &\quad  281,327, 379, 433, 489, 551, 615, 685, 757\dots.
    \end{aligned}
    \label{eq:growthour2}
\end{equation}
The generating function corresponding to the growth 
\eqref{eq:growthour2} is:
\begin{equation}
    g(z) = -\frac{2 z^9 - 2 z^8 - z^6 + z^5 + 8 z^3 + 2 z^2 + 7 z + 1}{%
        (z - 1)^3 (z^4 + z^3 + z^2 + z + 1)}.
    \label{eq:gfour2}
\end{equation}
All the poles of $g$ lie on the unit circle, so that the entropy
is zero.
Moreover, due to the presence of the factor $\left( z-1 \right)^{3}$
following remark \ref{rem:asydeg} we have that equation
\eqref{our2} has quadratic growth.

\subsubsection{Equation \eqref{ourii}}
Equation \eqref{ourii} is symmetric and has the following growth of degrees:
\begin{equation}
    1, 5, 7, 13, 19, 29, 37, 49, 61, 77, 91, 109, 127\dots.
    \label{eq:growthourii}
\end{equation}
The generating function corresponding to the growth
\eqref{eq:growthourii} is:
\begin{equation}
    g(z) = -\frac{z^5 - z^4 + 4 z^3 - 2 z^2 + 3 z + 1}{(z - 1)^3 (z + 1) (z^2 + 1)}.
    \label{eq:gfourii}
\end{equation}
All the poles of $g$ lie on the unit circle, so that the entropy
is zero.
Moreover, due to the presence of the factor $\left( z-1 \right)^{3}$
following remark \ref{rem:asydeg} we have that equation
\eqref{ourii} has quadratic growth.

\subsection{List 5}
\subsubsection{Equation \eqref{INB}}
Equation \eqref{INB} is symmetric and has the following growth of degrees:
\begin{equation}
    1, 2, 3, 4, 6, 8, 10, 13, 16, 19, 23, 27, 31\dots.
    \label{eq:growthINB}
\end{equation}
The generating function corresponding to the growth 
\eqref{eq:growthINB} is:
\begin{equation}
    g(z) = -\frac{z^4 - z^3 + 1}{(z - 1)^3 (z^2 + z + 1)}.
    \label{eq:gfINB}
\end{equation}
All the poles of $g$ lie on the unit circle, so that the entropy
is zero.
Moreover, due to the presence of the factor $\left( z-1 \right)^{3}$
following remark \ref{rem:asydeg} we have that equation
\eqref{INB} has quadratic growth.

\subsubsection{Equation \eqref{mod_INB}}
Equation \eqref{mod_INB} depends on the parameter $a$.
Using a simple scaling if $a\neq0$ it is possible to set $a=1$. 
For this reason we can consider the two cases $a=1$ and $a=0$.
If $a=1$ equation \eqref{mod_INB} is asymmetric, but it has
the following growth of degrees in both directions:
\begin{equation}
    \begin{aligned}
        &1, 2, 2, 4, 5, 7, 8, 11, 12, 16, 18, 22, 24, 30, 31, 
        \\
        &\quad38, 41, 47, 50, 59, 60, 70, 74, 82, 86, 98, 99\dots.
    \end{aligned}
    \label{eq:growthmod_INB}
\end{equation}
The generating function corresponding to the growth 
\eqref{eq:growthmod_INB} is:
\begin{equation}
    g(z) = -\frac{z^{13} + z^{10} + z^9 - z^7 + 2 z^5 + z^4 + z^3 + z^2 + 2 z + 1}{%
    (z - 1)^3 (z + 1)^2 (z^2 - z + 1) (z^2 + z + 1)^2}.
    \label{eq:gfmod_INB}
\end{equation}
If $a=0$ equation \eqref{mod_INB} is symmetric, but its growth of
degrees is still given by the sequence \eqref{eq:growthmod_INB}
and fitted by the generating function \eqref{eq:gfmod_INB}.
Therefore in both cases the entropy is zero since all the poles of 
$g$ lie on the unit circle.
Moreover, due to the presence of the factor $\left( z-1 \right)^{3}$
following remark \ref{rem:asydeg} we have that equation
\eqref{mod_INB} has quadratic growth for all values of $a$.
%
%

\subsubsection{Equation \eqref{mikh1}}
Equation \eqref{mikh1} depends on the parameter $a$.
Using a simple scaling if $a\neq0$ it is possible to set $a=1$. 
For this reason we can consider the two cases $a=1$ and $a=0$.
Equation \eqref{mikh1} is symmetric for both $a=1$ and $a=0$.
Moreover, in both cases it has the following growth of degrees:
\begin{equation}
    1, 4, 6, 10, 16, 22, 29, 37, 46, 56, 67, 79, 92\dots.
    \label{eq:growthmikh1}
\end{equation}
The generating function corresponding to the growth 
\eqref{eq:growthmikh1} is:
\begin{equation}
    g(z) = -\frac{(z^2 - z + 1) (z^4 - z^3 - 2 z^2 + 2 z + 1)}{(z - 1)^3}.
    \label{eq:gfmikh1}
\end{equation}
If $a=0$ equation \eqref{mikh1} is symmetric, but its growth of
degrees is still given by the sequence \eqref{eq:growthmikh1}
and fitted by the generating function \eqref{eq:gfmikh1}.
Therefore in both cases the entropy is zero since all the poles of 
$g$ lie on the unit circle.
Moreover, due to the presence of the factor $\left( z-1 \right)^{3}$
following remark \ref{rem:asydeg} we have that equation
\eqref{mikh1} has quadratic growth for all values of $a$.

%
%
%

\subsubsection{Equation \eqref{eq1}}
Equation \eqref{eq1} is symmetric and has the following growth of degrees:
\begin{equation}
    1, 4, 7, 10, 15, 21, 27, 36, 45, 54, 65, 77, 89, 104, 119, 134, 151, 169, 187, 208\dots.
    \label{eq:growtheq1}
\end{equation}
The generating function corresponding to the growth 
\eqref{eq:growtheq1} is:
\begin{equation}
    g(z) = -\frac{(z^4 + z + 1) (z^3 - z^2 + z + 1)}{(z - 1)^3 (z + 1) (z^2 - z + 1) (z^2 + z + 1)}.
    \label{eq:gfeq1}
\end{equation}
All the poles of $g$ lie on the unit circle, so that the entropy
is zero.
Moreover, due to the presence of the factor $\left( z-1 \right)^{3}$
following remark \ref{rem:asydeg} we have that equation
\eqref{eq1} has quadratic growth.

\subsubsection{Equation \eqref{eq311}}
Equation \eqref{eq311} is symmetric and has the following growth of degrees:
\begin{equation}
    \begin{aligned}
        &1, 3, 1, 4, 5, 5, 6, 13, 7, 15, 17, 17, 19, 31, 
        \\
        &\quad 21,34, 37, 37, 40, 57, 43, 61, 65, 65, 69\dots.
    \end{aligned}
    \label{eq:growtheq311}
\end{equation}
The generating function corresponding to the growth 
\eqref{eq:growtheq311} is:
\begin{equation}
    g(z) = -\frac{z^{10} + z^9 + z^7 - z^6 + z^5 + z^4 + 3 z + 1}{%
        (z - 1)^3 (z + 1)^2 (z^2 - z + 1) (z^2 + z + 1)^2}.
    \label{eq:gfeq311}
\end{equation}
All the poles of $g$ lie on the unit circle, so that the entropy
is zero.
Using the $\mathcal{Z}$-transform we obtain the following
expression for the degrees:
\begin{equation}
    \begin{aligned}
        d_{n} &= 
        \frac{n^{2}}{9}+\frac{5n}{9}+{\frac{191}{108}}
        +{\frac {5\left( -1 \right) ^{n}}{12}}+\frac{\left( -1\right) ^{n}n}{6}
        \\
        &+\frac{\sqrt {3}}{36}\sin \left( \frac{n\pi}{3} \right) 
        -\sqrt {3}\left(\frac{5n}{54}+\frac{7}{36}  \right)\sin \left( \frac{2n\pi}{3} \right)
        \\
        &+\frac{1}{12} \cos\left( \frac{n\pi}{3} \right) 
        + \left( {\frac {5 n}{18}}+{\frac{79}{108}}\right) \cos \left( \frac{2n\pi}{3} \right).
    \end{aligned}
    \label{eq:dneq311}
\end{equation}
Therefore the growth \eqref{eq:dneq311} is quadratic as
$n\to\infty$, but we notice also the unusual presence of oscillating
term proportional to $\left( -1 \right)^{n}n$ which explains the high 
oscillations of the sequence \eqref{eq:growtheq311}.
A similar was found in \cite{GSL_general} on the degree pattern of
some linearisable quad-equations.


\subsubsection{Equation \eqref{INB1}}
Equation \eqref{INB1} is symmetric and has the following growth of degrees:
\begin{equation}
    1, 3, 4, 7, 11, 15, 20, 25, 31, 38, 45, 53, 62, 71, 81, 92, 103, 115, 128, 141, 155\dots.
    \label{eq:growthINB1}
\end{equation}
The generating function corresponding to the growth 
\eqref{eq:growthINB1} is:
\begin{equation}
    g(z) = -\frac{z^8 - z^7 - z^6 + z^5 + z^3 - z^2 + z + 1}{(z - 1)^3 (z^2 + z + 1)}.
    \label{eq:gfINB1}
\end{equation}
All the poles of $g$ lie on the unit circle, so that the entropy
is zero.
Moreover, due to the presence of the factor $\left( z-1 \right)^{3}$
following remark \ref{rem:asydeg} we have that equation
\eqref{INB1} has quadratic growth.
%

\subsubsection{Equation \eqref{INB3}}
Equation \eqref{INB3} depends on the parameter $a$.
Using a simple scaling if $a\neq0$ it is possible to set $a=1$. 
For this reason we can consider the two cases $a=1$ and $a=0$.
Equation \eqref{INB3} is not symmetric for both $a=1$ and $a=0$.
However, in both cases it has the following growth of degrees:
\begin{equation}
    1, 3, 6, 9, 13, 19, 24, 31, 40, 48, 57, 69, 79, 91, 106, 119, 133, 151, 166, 183, 204\dots.
    \label{eq:growthINB3}
\end{equation}
The generating function corresponding to the growth 
\eqref{eq:growthINB3} is:
\begin{equation}
    g(z) = -\frac{z^9 + z^7 + z^6 + 3 z^5 + 2 z^4 + 2 z^3 + 3 z^2 + 2 z + 1}{%
        (z - 1)^3 (z + 1) (z^2 - z + 1) (z^2 + z + 1)^2)}.
    \label{eq:gfINB3}
\end{equation}
Therefore in both cases the entropy is zero since all the poles of 
$g$ lie on the unit circle.
Moreover, due to the presence of the factor $\left( z-1 \right)^{3}$
following remark \ref{rem:asydeg} we have that equation
\eqref{INB3} has quadratic growth for all values of $a$.
%

\subsubsection{Equation  \eqref{MX2}}
Equation \eqref{MX2} depends on the parameter $a$.
Using a simple scaling if $a\neq0$ it is possible to set $a=1$. 
For this reason we can consider the two cases $a=1$ and $a=0$.
Equation \eqref{MX2} is symmetric for both $a=1$ and $a=0$.
When $a=1$ equation \eqref{MX2} has the following growth of degrees:
\begin{equation}
    \begin{aligned}
        &1, 7, 15, 24, 35, 51, 66, 85, 109, 132, 157, 189, 218,
        \\
        &\quad 251, 291,328, 367, 415, 458, 505, 561, 612, 665\dots.
    \end{aligned}
    \label{eq:growthMX2a1}
\end{equation}
The generating function corresponding to the growth 
\eqref{eq:growthMX2a1} is:
\begin{equation}
    g_{a=1}(z) = -\frac{z^{10} + 2 z^7 + 5 z^6 + 8 z^5 + 5 z^4 + 8 z^3 + 8 z^2 + 6 z + 1}{
        (z - 1)^3 (z + 1) (z^2 - z + 1) (z^2 + z + 1)^2}.
    \label{eq:gfMX2a1}
\end{equation}
When $a=0$ equation \eqref{MX2} has the following growth of degrees:
\begin{equation}
    \begin{aligned}
        &1, 7, 15, 23, 33, 48, 63, 84, 107, 130, 155, 182, 211, 
        \\
        &\quad 248,287,324, 363, 404, 447, 500, 555, 606, 659, 
        \\
        &\quad  714,771, 840, 911, 976, 1043, 1112, 1183, 1268, 1355\dots.
    \end{aligned}
    \label{eq:growthMX2a0}
\end{equation}
The generating function corresponding to the growth 
\eqref{eq:growthMX2a1} is:
\begin{equation}
    g_{a=0}(z) = -\frac{\left(
        \begin{gathered}
            2 z^{16} - 3 z^{15} + 2 z^{14} - z^{13} + 2 z^{12} - 2 z^{11} - z^{10}+ 6 z^9
            \\
             + z^8 + 6 z^7 + 3 z^6 + 10 z^5 + 5 z^4 + 5 z^3 + 3 z^2 + 5 z + 1
     \end{gathered}\right)
        }{%
        (z - 1)^3 (z + 1) (z^2 - z + 1)^2 (z^2 + z + 1)^2}.
    \label{eq:gfMX2a0}
\end{equation}

All the poles of $g_{a=1}$ and $g_{a=0}$ lie on the unit circle, 
so that the entropy is zero in both cases.
Moreover, due to the presence of the factor $\left( z-1 \right)^{3}$
following remark \ref{rem:asydeg} we have that equation
\eqref{MX2} has quadratic growth in both cases.

\subsubsection{Equation \eqref{INB2}}
Equation \eqref{INB2} is symmetric and has the same growth of degrees
as equation \eqref{INB3}.
Therefore we have that equation \eqref{INB2} has zero entropy
and quadratic growth.
%

\subsection{List 6}
\subsubsection{Equation \eqref{seva}}
Equation \eqref{seva} is symmetric and has the following growth of degrees:
\begin{equation}
    1, 4, 6, 11, 16, 22, 29, 37, 46, 56, 67, 79, 92\dots.
    \label{eq:growthseva}
\end{equation}
The generating function corresponding to the growth 
\eqref{eq:growthseva} is:
\begin{equation}
    g(z) = -\frac{z^5 - 3 z^4 + 4 z^3 - 3 z^2 + z + 1}{(z - 1)^3}.
    \label{eq:gfseva}
\end{equation}
All the poles of $g$ lie on the unit circle, so that the entropy
is zero.
Moreover, due to the presence of the factor $\left( z-1 \right)^{3}$
following remark \ref{rem:asydeg} we have that equation
\eqref{seva} has quadratic growth.

\subsubsection{Equation \eqref{rat}}
Equation \eqref{rat} is symmetric and has the following growth of degrees:
\begin{equation}
    1, 6, 13, 25, 42, 61, 85, 111, 139, 171, 207, 245, 287, 333, 381, 433, 489, 547\dots.
    \label{eq:growthrat}
\end{equation}
The generating function corresponding to the growth 
\eqref{eq:growthrat} is:
\begin{equation}
    g(z) = -\frac{2 z^{10} - z^9 - 3 z^7 + z^4 + 4 z^3 + 2 z^2 + 4 z + 1}{(z - 1)^3 (z^2 + z + 1)}.
    \label{eq:gfrat}
\end{equation}
All the poles of $g$ lie on the unit circle, so that the entropy
is zero.
Moreover, due to the presence of the factor $\left( z-1 \right)^{3}$
following remark \ref{rem:asydeg} we have that equation
\eqref{rat} has quadratic growth.

\subsubsection{Equation \eqref{sroot}}
We proved that equation \eqref{sroot} can be brought in
in rational form \eqref{eq:hyp}, but this form is not
bi-rational.
So we cannot apply the algebraic entropy method to
this equation.

\subsubsection{Equation \eqref{SK2}}
Equation \eqref{SK2} is symmetric and has the following growth of degrees:
\begin{equation}
    1, 7, 15, 24, 35, 49, 67, 86, 107, 132, 159, 188, 219, 254, 291, 330, 371, 416\dots.
    \label{eq:growthSK2}
\end{equation}
The generating function corresponding to the growth 
\eqref{eq:growthSK2} is:
\begin{equation}
    g(z) = -\frac{z^{11} - 2 z^{10} + z^9 + 2 z^6 - 2 z^5 + z^4 + z^3 + 2 z^2 + 5 z + 1}{(z - 1)^3 (z + 1) (z^2 + 1)}.
    \label{eq:gfSK2}
\end{equation}
All the poles of $g$ lie on the unit circle, so that the entropy
is zero.
Moreover, due to the presence of the factor $\left( z-1 \right)^{3}$
following remark \ref{rem:asydeg} we have that equation
\eqref{SK2} has quadratic growth.

\section{Discussion}
\label{sec:disc}

In the previous section we computed the algebraic entropy of all the 
integrable Volterra-like five-point differential-difference equations recently
classified in \cite{GarifullinYamilovLevi2016,GarifullinYamilovLevi2018}.
When possible, we showed that the method of algebraic entropy and the method
of generalised symmetries agree.
That is, we showed that all the equations integrable according to the
generalised symmetry test are also integrable according to the algebraic
entropy method, i.e. the algebraic entropy is zero.
The algebraic entropy method is unfortunately unable to treat the
semi-discrete Kaup-Kaupershmidt equation \eqref{sroot}.
This is because the generalised symmetry approach, differently from
the algebraic entropy, makes no assumption on the nature of the recurrence
and algebraic or even transcendental terms are allowed.
That is, we proved that the for integrable Volterra-like five-point 
differential-difference equations the following version of the
algebraic entropy conjecture holds true:
\begin{conjecture}
    The condition that algebraic entropy is zero is equivalent to
    the definition of integrability for bi-rational maps.
\end{conjecture}

Except for the two known non-trivially linearisable equations
\eqref{Bur2} and \eqref{Bur} the growth is always quadratic. 
Equation \eqref{eq311} possesses an interesting non-standard
highly oscillating growth, observed for the first time in 
differential-difference equations.
Nevertheless the asymptotic growth is still quadratic.

In the case of two-dimensional difference equation it is known
that the only possible polynomial, i.e. integrable, growth is
quadratic \cite{Diller1996}.
Integrable higher order maps can exhibit higher rate of growth,
see e.g. \cite{JoshiViallet2017,GJTV_class,GJTV_sanya}.
Despite being infinite-dimensional all the integrable Volterra-like 
five-point differential-difference equations possess this ``minimal''
integrable growth.

In the case of difference equations it has been observed that degree
growth greater than quadratic is related to a procedure called 
\emph{deflation} \cite{JoshiViallet2017}.
That is, a five-point equation is reduced to a four-point one
using a non-point potential-like transformations of the form:
\begin{equation}
    v_{n} = \frac{a_{1}u_{n}u_{n+1}+a_{2}u_{n+1}+a_{3}u_{n}+a_{4} }{%
        b_{1}u_{n}u_{n+1}+b_{2}u_{n+1}+b_{3}u_{n}+b_{4}}.
    \label{eq:defl}
\end{equation}
Let us notice, that the inflated version of a Volterra-like
differential-difference equation is not always a Volterra-like
differential-difference equation.
This fact makes more difficult to make predictions on the integrability
properties of the inflated forms of differential-difference equations.

Finally, we notice that another interesting problem is to study the
integrability properties of the stationary reductions of the integrable
Volterra-like five-point differential-difference equations.
The stationary reduction of a bi-rational five-point differential-difference
equation is a fourth-order difference equation, i.e. a four-dimensional
map of the projective space into itself.
It will be important to understand how integrability arises
inside these families of equations and if it fits with known
cases of integrable families of fourth-order differential difference
equations \cite{GJTV_class,GJTV_sanya,CapelSahadevan2001}.
%

\section*{Acknowledgment} 

GG thanks Prof. R. N. Garifullin, Prof. D. Levi and Prof. R. I. Yamilov for
interesting and helpful discussions during the preparation of this paper.

GG is supported by the Australian Research Council through Nalini Joshi's
Australian Laureate Fellowship grant FL120100094.

\bibliographystyle{plain}
\bibliography{bibliography}

\end{document}